%% file: paper_draft_en.tex
\documentclass[a4paper,11pt]{article}
\usepackage[margin=25mm]{geometry}
\usepackage{amsmath,amssymb,amsthm}
\usepackage{etoolbox}
\usepackage{booktabs}
\usepackage{graphicx}
\usepackage{hyperref}
\usepackage{bookmark}
\hypersetup{hidelinks,unicode=true,pdfencoding=auto,psdextra,pdfauthor={Koki Okada; Kenta Kasai}}
\newcommand{\vect}[1]{\boldsymbol{#1}}

\newcommand{\addqedtoenv}[1]{%
  \AtBeginEnvironment{#1}{\pushQED{\qed}}%
  \AtEndEnvironment{#1}{\popQED}%
}
\theoremstyle{definition}
\newtheorem{theorem}{Theorem}
\newtheorem{lemma}[theorem]{Lemma}

\newtheorem{corollary}[theorem]{Corollary}
\newtheorem{definition}[theorem]{Definition}

\addqedtoenv{theorem}
\addqedtoenv{lemma}
\addqedtoenv{proposition}
\addqedtoenv{corollary}
\addqedtoenv{definition}
\addqedtoenv{remark}
\makeatletter
\renewenvironment{proof}[1][\proofname]{%
  \par\pushQED{\qed}\normalfont
  \topsep6\p@\@plus6\p@\relax
  \trivlist
  \item[\hskip\labelsep\bfseries #1\@addpunct{.}]\ignorespaces
}{%
  \popQED\endtrivlist\@endpefalse
}
\makeatother

\title{High-Girth Regular Quantum LDPC Codes from Square-Base Hypergraph Products via CPM Lifts}
\author{Koki Okada \qquad Kenta Kasai\\Institute of Science Tokyo\\Email: \texttt{kenta@ict.eng.isct.ac.jp}}
\date{}

\begin{document}
\maketitle

\begin{abstract}
We study square-base Calderbank--Shor--Steane (CSS) hypergraph-product codes as a finite-length class for regular high-girth quantum low-density parity-check (LDPC) design.
For base matrices of small column weight, we give checkable conditions for regularity, rank deficiency, and short-cycle exclusion, and we present explicit column-weight-three and column-weight-four examples with Tanner girth 6 and 8.
We also analyze circulant permutation matrix (CPM) lifts of this class.
Using the standard voltage-sum criterion, we identify orthogonality-forced Tanner 8-cycles and show that CPM lifting cannot raise the Tanner girth beyond 8 when these cycles are present.
As a representative finite-length instance, a randomized CPM lift of the girth-8 base construction gives a $[[28800,62,18\le d\le192]]$ girth-8 $(3,6)$-regular CSS-LDPC code. Complete enumeration excludes nontrivial logical supports of weight at most $16$ on both sides, giving $d_X,d_Z\ge18$; explicit weight-$192$ representatives give $d_X,d_Z\le192$.
Under degeneracy-aware belief-propagation decoding with optional ordered-statistics-decoding-lite post-processing, this code produced zero decoding failures in $2.993\times 10^8$ independent trials at depolarizing probability $p=0.1402$; the Wilson 95\% upper confidence bound is $1.28\times 10^{-8}$.
\end{abstract}


\section{Introduction}

Classical low-density parity-check (LDPC) codes were introduced as sparse parity-check codes with low-complexity message-passing decoding~\cite{Gallager1962,Tanner1981}. Subsequent analysis of belief propagation (BP) decoding made small-column-weight regular LDPC codes a baseline class: the degree structure is fixed, the Tanner graph is sparse, and the decoder is implemented by local message updates.

Classical LDPC design also developed through optimized irregular Tanner graphs. Density evolution made the BP threshold a computable function of the degree distribution, and carefully optimized irregular ensembles were shown to operate close to channel capacity~\cite{Luby2001Irregular,RichardsonUrbanke2001Capacity,ChungRichardsonUrbanke2001Design}. These results treat the degree distribution itself as a design object; an unsuitable distribution can have a poor threshold or poor finite-length behavior.

Quantum LDPC codes share the sparse-check objective but have an additional structural constraint. Calderbank--Shor--Steane (CSS) codes represent stabilizer codes by two classical parity-check matrices~\cite{CalderbankShor1996,Steane1996}. These matrices must commute, so the choices of row and column degrees cannot be varied independently in the same way as in a classical irregular LDPC ensemble. Sparse-graph quantum codes, iterative decoding, and the resulting design constraints are reviewed and developed in earlier work~\cite{MacKayMitchisonMcFadden2004,PoulinChung2008Iterative,BreuckmannEberhardt2021QLDPCSurvey}.

Product-based constructions provide a way to satisfy the commutation constraint by deriving both CSS check matrices from a common algebraic or combinatorial object. The hypergraph-product (HGP) construction gives quantum LDPC codes with positive asymptotic rate from classical codes~\cite{TillichZemor2014}. Later product-code and intersecting-subset constructions give additional CSS families in which commutation follows from an incidence or product structure~\cite{OstrevOrsucciLazaroMatuz2024,Ostrev2024IntersectingSubsets}.

Two recent works provide the immediate technical context for the present study. The first formulated finite-length constraints that appear when one requires regular quantum LDPC codes with large Tanner girth while maintaining CSS orthogonality, and used controlled permutation structure to exclude specified low-weight structures~\cite{Kasai2026OrthogonalityBarrier}. The second gave an affine-coset construction of a regular high-girth CSS code, applied CPM lifts, and evaluated the lifted codes with a BP decoder and post-processing under depolarizing noise~\cite{OkadaKasai2026AffineCoset}.

This paper studies a restricted family: square-base CSS-HGP codes with small column weight. Within this class, the verification tasks are explicit: determine the degree conditions that give regularity, compute the rank deficiencies that give a positive number of logical qubits, identify the short cycles forced by CSS orthogonality after CPM lifting, and measure the finite-length decoding behavior of selected lifted constructions. Random search can produce additional finite examples, but the examples below are chosen mainly to show how known incidence structures and simple combinatorial modifications realize the required checks. The column-weight-three examples are a Fano-plane incidence matrix with Tanner girth 6, a generalized-quadrangle incidence matrix with Tanner girth 8, and comparison matrices with the same target degree. Column-weight-four finite-geometric examples are also included to show that the same square-base analysis covers the $(4,8)$-regular case. The decoding experiments use CPM lifts of the column-weight-three girth-8 generalized-quadrangle construction.

\section{General Construction Theory}

This section separates quantities belonging to the base HGP code from those belonging to the CPM-lifted code. We call the binary matrix used as the input to the HGP construction the base matrix. We use unmarked symbols for the base HGP code and hats for the CPM-lifted code.

The basic construction used in this manuscript is the HGP construction introduced in~\cite{TillichZemor2014}, specialized to a square binary base matrix. Let the base-matrix size be $s$ and take
\[
B\in \mathbb{F}_2^{s\times s}.
\]
All vectors are treated as column vectors and are written in boldface. For example, $\vect{x}$ is a column vector, and its transpose is written as $\vect{x}^{\mathsf T}$. The base-HGP CSS check matrices are
\[
H_X(B)= [B \otimes I_s \mid I_s \otimes B^{\mathsf T}], \qquad
H_Z(B)= [I_s \otimes B \mid B^{\mathsf T} \otimes I_s].
\]
They always satisfy the CSS orthogonality condition because
\[
H_X(B)H_Z(B)^{\mathsf T}
=B\otimes B^{\mathsf T}+B\otimes B^{\mathsf T}=0
\]
over $\mathbb{F}_2$.

The number of logical qubits is controlled by the rank deficiency of the base matrix. The following definition fixes the terminology used in the parameter formulas below.

\begin{definition}[Corank]
For a square base matrix $B\in\mathbb{F}_2^{s\times s}$, set
\[
\rho_B=\mathrm{rank}_{\mathbb{F}_2}(B),\qquad
c_B=\dim_{\mathbb{F}_2}\ker B=s-\rho_B
\]
and call $c_B$ the corank of $B$. Thus the corank is the rank deficiency of the base matrix, equivalently the dimension of its kernel over $\mathbb{F}_2$.
\end{definition}

For a CSS code specified by length $n$, an $m_X$-row matrix $H_X$, and an $m_Z$-row matrix $H_Z$, define the design number of logical qubits and design quantum rate by
\[
k_{\mathrm{des}}=n-m_X-m_Z,\qquad
R_{\mathrm{des}}=\frac{k_{\mathrm{des}}}{n}.
\]
These are design-level quantities obtained by assuming that all check rows are independent. They are distinct from the actual number of logical qubits
\[
k=n-\mathrm{rank}(H_X)-\mathrm{rank}(H_Z)
\]
and the actual rate $R=k/n$.

Assume that $B$ is a square $w_B$-regular base matrix. In the base HGP code, $n$ denotes the number of physical qubits, $m_X$ and $m_Z$ denote the numbers of $X$- and $Z$-check rows, $\rho_X$ and $\rho_Z$ denote the ranks of $H_X(B)$ and $H_Z(B)$ over $\mathbb F_2$, and $k$ denotes the number of logical qubits. With $s$ the base-matrix size and $c_B$ the corank of $B$, these parameters are
\[
n=2s^2,\qquad
m_X=m_Z=s^2,\qquad
\rho_X=\rho_Z=s^2-c_B^2,\qquad
k=2c_B^2.
\]
Indeed, the square-HGP formula for $k$ gives
\[
\rho_X=\mathrm{rank}(H_X)=s^2-c_B^2,\qquad
\rho_Z=\mathrm{rank}(H_Z)=s^2-c_B^2,
\]
and hence $k=n-\rho_X-\rho_Z=2c_B^2$. This is the standard rank formula for the HGP construction~\cite{TillichZemor2014}.

On the other hand, this square-HGP construction has $n=m_X+m_Z$, so
\[
k_{\mathrm{des}}=0,\qquad R_{\mathrm{des}}=0.
\]
Thus, the positive values of $k$ considered in this paper do not come from a positive design quantum rate; they come from rank deficiencies induced by the corank of the base matrix. A CPM $P$ lift multiplies $n,m_X,m_Z$ by the same factor $P$, so the lifted design quantum rate also remains zero.

The next theorem isolates how regularity behaves under this construction. It shows that regularity of the quantum LDPC code is not an additional property imposed after the HGP step; in this setting it is equivalent to row-and-column regularity of the square base matrix.

\begin{theorem}[Regularity condition]
\label{thm:regularity-condition}
Let $a_i$ be the weight of row $i$ of $B$, and let $b_j$ be the weight of column $j$. Then the check row $(i,j)$ of $H_X(B)$ has weight $a_i+b_j$, while the check row $(i,j)$ of $H_Z(B)$ has weight $b_i+a_j$. Moreover, column weights in the first variable block are column weights of $B$, and column weights in the second variable block are row weights of $B$.

Consequently, this HGP construction gives a regular CSS code with variable degree $d_{\mathrm v}$ and check degree $d_{\mathrm c}$ if and only if $B$ is a square regular base matrix whose rows and columns all have weight $d_{\mathrm v}$. In that case
\[
d_{\mathrm c}=2d_{\mathrm v}.
\]
Furthermore, any CPM $P$ lift obtained by replacing nonzero entries with circulant permutation blocks preserves these degrees, so the lifted code satisfies $\hat{d}_{\mathrm v}=d_{\mathrm v}$ and $\hat{d}_{\mathrm c}=d_{\mathrm c}$.
\end{theorem}

\begin{proof}
The check row $(i,j)$ of $H_X(B)$ receives $a_i$ nonzero entries from row $i$ of $B$ in the first block and $b_j$ nonzero entries from column $j$ of $B$ in the second block. Hence its row weight is $a_i+b_j$. The same argument for $H_Z(B)$ gives row weight $b_i+a_j$, since row $j$ of $B$ contributes in one block and column $i$ contributes in the other.

For column weights, columns in the first block of $H_X(B)$ correspond to columns of $B$, while columns in the second block correspond to rows of $B$. The same two sets of weights occur for $H_Z(B)$. Therefore all variable columns have the same weight $d_{\mathrm v}$ if and only if all row weights and all column weights of $B$ are equal to $d_{\mathrm v}$. In that case every check row has weight $d_{\mathrm v}+d_{\mathrm v}=2d_{\mathrm v}$. Conversely, if the HGP code is regular, then the column weights in the first and second variable blocks must be equal, forcing all row and column weights of $B$ to be the same value $d_{\mathrm v}$.

Finally, a CPM lift only replaces each 1 by a circulant permutation matrix of the same size. Each lifted check row and lifted variable column therefore sees the same number of nonzero blocks as in the base matrix, so the degrees are preserved.
\end{proof}

Regularity alone does not control short cycles in the Tanner graph. The following elementary lemma records the support-overlap calculation used as a design filter for HGP matrices. The condition is not necessary; it is included to state a directly checkable sufficient condition for excluding 4-cycles and 6-cycles after the HGP step.

\begin{lemma}[A sufficient condition for short-cycle removal]
\label{lem:short-cycle-removal}
If the Tanner graph of $B\in\mathbb{F}_2^{s\times s}$ has no 4-cycles and no simple 6-cycles, then the Tanner graphs of $H_X(B)$ and $H_Z(B)$ also have no 4-cycles and no simple 6-cycles.
\end{lemma}

\begin{proof}
We use the standard Tanner-graph common-neighbor characterization of short cycles in sparse parity-check matrices~\cite{Tanner1981}. In particular, the Tanner graph of $B$ has no 4-cycles if and only if any two distinct rows and any two distinct columns of $B$ share at most one common 1. By assumption, the Tanner graph of $B$ also has no simple 6-cycles.

It suffices to prove the claim for $H_X(B)$; the proof for $H_Z(B)$ is obtained by interchanging rows and columns. Let $R_i$ be the support of row $i$ of $B$, and let $C_j$ be the support of column $j$. Index the check rows of $H_X(B)$ by pairs $(i,j)$. The support of the check row $(i,j)$ is
\[
\{(a,j)\mid a\in R_i\}
\cup
\{(i,b)^\star\mid b\in C_j\},
\]
where $\star$ denotes variables in the second block. Hence the number of common variables between two distinct check rows $(i,j)$ and $(i',j')$ is
\[
\delta_{j,j'}|R_i\cap R_{i'}|
+
\delta_{i,i'}|C_j\cap C_{j'}|.
\]
By the assumption that the Tanner graph of $B$ has no 4-cycles, this number is at most 1. Thus no pair of checks shares two or more variables, so the Tanner graph has no 4-cycles.

Suppose next that the Tanner graph of $H_X(B)$ contains a simple 6-cycle. Such a cycle contains three distinct check rows, and each pair of them shares exactly one variable, with the three shared variables distinct. From the overlap formula above, two check rows can share a variable only if their first coordinates agree or their second coordinates agree. For three distinct coordinate pairs to be pairwise adjacent in this way, either all first coordinates are equal or all second coordinates are equal.

In the latter case the three check rows can be written as $(i_1,j),(i_2,j),(i_3,j)$. Since the three shared variables in the 6-cycle are distinct, the Tanner graph of $B$ contains the simple 6-cycle
\[
i_1-a_{12}-i_2-a_{23}-i_3-a_{31}-i_1,
\]
contradicting the assumption that the Tanner graph of $B$ has no simple 6-cycles. The case where all first coordinates are equal gives the same contradiction on the column side. Therefore $H_X(B)$ has no 6-cycles.
\end{proof}

This elementary criterion gives the base-matrix selection rule used below: exclude row-pair intersections of size at least 2 and simple 6-cycles while keeping the base-matrix corank positive and within the desired parameter range.

A CPM $P$ lift replaces each nonzero entry of the base matrices $H_X,H_Z$ by a $P\times P$ circulant permutation block, producing lifted matrices $\hat{H}_X,\hat{H}_Z$. This circulant-permutation representation can be interpreted as a voltage-graph cover, and the corresponding cycle congruence conditions are standard for circulant-permutation LDPC constructions~\cite{GrossTucker1987TopologicalGraphTheory,Fossorier2004CirculantPermutation}. The lift is a graph cover of each Tanner graph, so
\[
\hat{n}=P n,\qquad
\hat{m}_X=P m_X,\qquad
\hat{m}_Z=P m_Z
\]
and check row weights and variable column weights are preserved. CSS orthogonality, however, is not automatic for arbitrary shifts: the signed sum of shifts along every relevant $X$/$Z$ overlap must vanish modulo $P$, giving a system of linear congruence constraints. A base cycle closes in the lift only when its signed shift sum is zero modulo $P$, so short-cycle suppression becomes the problem of satisfying nonzero congruence conditions for short cycles while preserving the zero congruences required by CSS orthogonality. We write the lifted ranks and number of logical qubits as
\[
\hat{\rho}_X=\mathrm{rank}(\hat{H}_X),\qquad
\hat{\rho}_Z=\mathrm{rank}(\hat{H}_Z),\qquad
\hat{k}=\hat{n}-\hat{\rho}_X-\hat{\rho}_Z,
\]
so $\hat{k}$ is not determined by the base-matrix corank $c_B$ alone and $\hat{\rho}_X,\hat{\rho}_Z$ must be computed for each lift.

The CSS zero congruences can also force some base cycles to close in every lift. This is a direct application of the standard voltage-sum criterion for graph covers and circulant-permutation LDPC constructions~\cite{GrossTucker1987TopologicalGraphTheory,Fossorier2004CirculantPermutation}. The following lemma states the local CSS pattern used later. The statement is conditional on a support pattern: CSS orthogonality alone does not imply that every CSS code has such an 8-cycle, but once the indicated local pattern is present and the corresponding overlaps are enforced by pairwise zero constraints, the cycle voltage is forced to vanish for every lift size.

\begin{lemma}[CSS-orthogonality-forced Tanner 8-cycles]
Let $H_X,H_Z$ be binary CSS check matrices with $H_XH_Z^{\mathsf T}=0$, and consider a CPM $P$ lift specified by shift variables
\[
s^X_{x,c},\qquad s^Z_{z,c}\in\mathbb Z_P
\]
for nonzero entries of $H_X$ and $H_Z$. Suppose that there are a $Z$-check $z$, four distinct $X$-checks $x_0,x_1,x_2,x_3$, and four distinct variables $c_0,c_1,c_2,c_3$ such that, with indices taken modulo $4$,
\[
\{c_{a-1},c_a\}\subseteq \mathrm{supp}(x_a)\cap \mathrm{supp}(z)
\qquad (a=0,1,2,3).
\]
Then the Tanner graph of $H_X$ contains the 8-cycle
\[
x_0-c_0-x_1-c_1-x_2-c_2-x_3-c_3-x_0.
\]
Assume moreover that, for each $a$, the two indicated variables $c_{a-1},c_a$ are paired in the CSS zero constraint for the pair $(x_a,z)$; in particular, this holds when $\mathrm{supp}(x_a)\cap \mathrm{supp}(z)=\{c_{a-1},c_a\}$. Then the voltage of this 8-cycle is $0\pmod P$. Hence the cycle closes in every CPM $P$ lift satisfying these CSS zero constraints. The same statement holds with $X$ and $Z$ interchanged.
\end{lemma}

\begin{proof}
The support assumptions give exactly the adjacencies used in the displayed cycle. Since the four checks $x_0,x_1,x_2,x_3$ and the four variables $c_0,c_1,c_2,c_3$ are distinct, this is a Tanner 8-cycle in the base Tanner graph of $H_X$.

It remains to compute its lift voltage. This is the standard voltage-sum calculation for graph covers and circulant-permutation LDPC constructions~\cite{GrossTucker1987TopologicalGraphTheory,Fossorier2004CirculantPermutation}. For each $a$, the pairwise CSS zero constraint for $(x_a,z)$ and the pair of shared variables $(c_{a-1},c_a)$ is
\[
s^X_{x_a,c_{a-1}}-s^Z_{z,c_{a-1}}
-s^X_{x_a,c_a}+s^Z_{z,c_a}
\equiv 0\pmod P .
\]
Summing these four congruences over $a=0,1,2,3$, all terms containing $s^Z_{z,c_a}$ cancel because each appears once with sign $+$ and once with sign $-$. The remaining expression is
\[
s^X_{x_0,c_3}-s^X_{x_0,c_0}
+s^X_{x_1,c_0}-s^X_{x_1,c_1}
+s^X_{x_2,c_1}-s^X_{x_2,c_2}
+s^X_{x_3,c_2}-s^X_{x_3,c_3}
\equiv 0\pmod P .
\]
Up to reversing the orientation, this is the signed voltage sum around
$x_0-c_0-x_1-c_1-x_2-c_2-x_3-c_3-x_0$. Therefore the cycle voltage is zero modulo $P$, so the 8-cycle lifts to closed 8-cycles. The derivation uses only an integer sum of CSS zero constraints, and hence it is independent of the value of $P$. Swapping $X$ and $Z$ gives the dual statement.
\end{proof}

The preceding lemma is a CSS-theoretic local statement obtained from the standard voltage criterion. In the square-HGP construction, the required local pattern appears systematically from the row and column overlaps of the base matrix $B$. The next theorem records this specialization and gives the explicit cycles used later.

\begin{theorem}[Square-HGP specialization of forced Tanner 8-cycles]
\label{thm:square-hgp-forced-8-cycles}
Let $B\in\mathbb{F}_2^{s\times s}$. Let $R_i$ be the support of row $i$ of $B$, and let $C_j$ be the support of column $j$. Suppose that for distinct rows $i\ne i'$ and distinct columns $j\ne j'$ there exist
\[
u\in R_i\cap R_{i'},\qquad v\in C_j\cap C_{j'}.
\]
Then the Tanner graph of $H_X(B)$ contains the Tanner 8-cycle
\[
(i,j)\;-\;(u,j)\;-\;(i',j)\;-\;(i',v)^\star\;-\;(i',j')\;-\;(u,j')\;-\;(i,j')\;-\;(i,v)^\star\;-\;(i,j),
\]
where $\star$ denotes the second variable block.

Moreover, if the shifts of a CPM $P$ lift satisfy the CSS orthogonality zero constraints, then the voltage sum of this 8-cycle is $0\pmod P$ for every $P\ge 1$. Hence this cycle closes in every CPM $P$ lift. The same statement holds for $H_Z(B)$ after swapping rows and columns.
\end{theorem}

\begin{proof}
The displayed vertices form a Tanner 8-cycle of $H_X(B)$ because $u\in R_i\cap R_{i'}$ and $v\in C_j\cap C_{j'}$.

It remains to show that the cycle closes after lifting, using the same voltage-sum criterion for covers~\cite{GrossTucker1987TopologicalGraphTheory,Fossorier2004CirculantPermutation}. Write $x(\cdot,\cdot)$ for $X$-side shifts and $z(\cdot,\cdot)$ for $Z$-side shifts. The $X$ check $(a,b)$ and the $Z$ check $(u,v)$ overlap in the two variables $(u,b)$ and $(a,v)^\star$ whenever $u\in R_a$ and $v\in C_b$. The CSS orthogonality zero constraint for this pair is
\[
x((a,b),(u,b))-z((u,v),(u,b))
-x((a,b),(a,v)^\star)+z((u,v),(a,v)^\star)\equiv 0\pmod P.
\]

Apply this congruence to the four choices $(i,j)$, $(i',j)$, $(i',j')$, and $(i,j')$ of $(a,b)$, with signs $+,-,+,-$. All terms containing $z$ cancel. The remaining $x$ terms are exactly the signed voltage sum around the Tanner 8-cycle displayed above. Therefore this voltage sum is $0\pmod P$. Since the argument uses only integer addition and subtraction of the orthogonality equations, it is independent of the value of $P$. The proof for $H_Z(B)$ is the same after swapping the $X$ and $Z$ roles and interchanging rows and columns.
\end{proof}

\begin{corollary}[Unavoidable 8-cycle count for regular square base matrices]
\label{cor:unavoidable-8-cycle-count}
Let $B$ be an $s\times s$ $w$-regular base matrix. In every CPM $P$ lift, each of the two Tanner graphs, for $H_X$ and for $H_Z$, contains at least
\[
N_8(B)=s^2\binom{w}{2}^2
\]
Tanner 8-cycles forced by orthogonality. In particular, if the Tanner graphs of $H_X(B)$ and $H_Z(B)$ have no 4-cycles and no 6-cycles, then the Tanner graphs of the lifted matrices $\hat{H}_X$ and $\hat{H}_Z$ have girth $8$ for every CPM $P$ lift; increasing $P$ cannot produce girth at least $10$ for these two Tanner graphs in this lift model.
\end{corollary}

\begin{proof}
Since $B$ is $w$-regular, the number of choices of a row pair together with a shared column is
\[
\sum_{u=1}^{s}\binom{w}{2}=s\binom{w}{2}.
\]
The number of choices of a column pair together with a shared row is the same. Theorem~\ref{thm:square-hgp-forced-8-cycles} gives one forced Tanner 8-cycle from each pair of such choices, hence $N_8(B)$ cycles on the $H_X$ side and the same number on the $H_Z$ side.

A CPM lift is a cover of the Tanner graph. By the standard graph-cover projection property, a cycle in the lift projects to a closed walk in the base graph~\cite{GrossTucker1987TopologicalGraphTheory}. Indeed, the projection of a simple cycle in the cover is a non-backtracking closed walk in the base graph, and every non-backtracking closed walk contains a simple cycle of length no larger than the walk length. Hence, if the Tanner graphs of the base HGP matrices have no 4-cycles and no 6-cycles, the lift cannot create new 4-cycles or 6-cycles in the Tanner graphs of $\hat{H}_X$ or $\hat{H}_Z$. Theorem~\ref{thm:square-hgp-forced-8-cycles} shows that 8-cycles remain for every $P$, so each of these two Tanner graphs has girth exactly $8$.
\end{proof}

The next theorem records the rank deficiency that cannot be lost under a permutation lift. In the square HGP codes considered in this paper, it implies that the number of logical qubits in the base code is a lower bound on the number of logical qubits after lifting.

\begin{theorem}[Logical-qubit lower bound after CPM lifting]
\label{thm:cpm-lift-k-lower-bound}
Let a base CSS code be specified by binary matrices $H_X,H_Z$ of length $n$ with $m_X,m_Z$ check rows, and let
\[
k=n-\mathrm{rank}(H_X)-\mathrm{rank}(H_Z)
\]
be its number of logical qubits. Let $\hat{H}_X,\hat{H}_Z$ be a $P$ lift obtained by replacing each nonzero entry by a $P\times P$ permutation matrix, and assume that the lifted CSS orthogonality condition is satisfied. Then
\[
\hat{k} \ge k+(P-1)(n-m_X-m_Z).
\]
In particular, if $n=m_X+m_Z$, as in the square HGP codes used in this paper, then
\[
\hat{k}\ge k.
\]
\end{theorem}

\begin{proof}
Set $\ell_X=m_X-\mathrm{rank}(H_X)$ and $\ell_Z=m_Z-\mathrm{rank}(H_Z)$. The number $\ell_X$ is the dimension of the space of row dependencies of $H_X$, equivalently the space of column vectors $\vect{u}$ satisfying $H_X^{\mathsf T}\vect{u}=\vect{0}$. Given such a row dependency $\vect{u}$, define the lifted column vector $\hat{\vect{u}}$ by making its $i$th block equal to $u_i\vect{1}_P$, where $\vect{1}_P$ is the all-one column vector of length $P$.

Every permutation matrix preserves $\vect{1}_P$. This is the constant-on-fibers subspace associated with a permutation cover, and it is the same elementary mechanism used in circulant-permutation lifts~\cite{GrossTucker1987TopologicalGraphTheory,Fossorier2004CirculantPermutation}. Therefore, whenever an even number of base rows cancels in a base column, the corresponding lifted block rows also cancel in that column block. Hence $\hat{H}_X^{\mathsf T}\hat{\vect{u}}=\vect{0}$. The map $\vect{u}\mapsto \hat{\vect{u}}$ is injective, so the lifted left-kernel dimension satisfies $\hat{\ell}_X\ge \ell_X$. The same argument gives $\hat{\ell}_Z\ge \ell_Z$.

Thus
\[
\mathrm{rank}(\hat{H}_X)=P m_X-\hat{\ell}_X,\qquad
\mathrm{rank}(\hat{H}_Z)=P m_Z-\hat{\ell}_Z,
\]
and
\[
\begin{aligned}
\hat{k}
&=Pn-(P m_X-\hat{\ell}_X)-(P m_Z-\hat{\ell}_Z)\\
&=P(n-m_X-m_Z)+\hat{\ell}_X+\hat{\ell}_Z\\
&\ge P(n-m_X-m_Z)+\ell_X+\ell_Z\\
&=k+(P-1)(n-m_X-m_Z).
\end{aligned}
\]
When $n=m_X+m_Z$, this reduces to $\hat{k}\ge k$.
\end{proof}

\section{Base-Matrix Design}

Theorem~\ref{thm:regularity-condition}, Lemma~\ref{lem:short-cycle-removal}, and the square-HGP parameter formulas above show that the object to be designed is not the HGP matrix after expansion, but the square binary input matrix $B$. The row and column weights of $B$ determine the degrees after the HGP step, the corank of $B$ determines the number of logical qubits in the base code, and the overlap structure of $B$ controls short cycles after the HGP step. We therefore use a general target column weight $w$ and design $B$ according to the following criteria.

\begin{enumerate}
\item The matrix $B$ is square and binary. This allows the square-HGP formulas above to be used directly and gives $n=m_X+m_Z$ for the base code. Hence the design quantum rate of this design is $R_{\mathrm{des}}=0$.
\item Choose a target weight $w$ and make all rows and columns of $B$ have weight $w$. The resulting HGP CSS code is then $(w,2w)$-regular, with variable column weight $w$ and check row weight $2w$. The main concrete examples use $w=3$, and $w=4$ is also included.
\item The corank of $B$ is positive and chosen according to the desired value of $k=2c_B^2$. If the corank is zero, the base code has no logical qubits; larger corank gives a larger base value of $k$.
\item To exclude 4-cycles and 6-cycles, any two distinct rows and any two distinct columns share at most one common 1, and the Tanner graph of $B$ has no simple 6-cycles. These are sufficient conditions for excluding 4-cycles and 6-cycles after the HGP step. For the girth-6 example, we keep the overlap condition but allow simple 6-cycles.
\item The base-matrix size is kept finite and explicit, and length is increased by CPM lifting when needed. The lifted value of $k$ still has to be computed for each lift, but Theorem~\ref{thm:cpm-lift-k-lower-bound} ensures that the base value of $k$ remains a lower bound for square HGP codes.
\end{enumerate}

These are not necessary conditions; they are checkable design conditions used in this paper. Randomized local search can also find matrices satisfying them, but the main examples are presented through known finite-geometric incidence structures and simple combinatorial modifications because their regularity, intersection properties, and rank behavior are easier to verify. We use five column-weight-three base matrices and two column-weight-four finite-geometric base matrices. The Fano-plane incidence matrix is denoted by $B_7$ and gives a girth-6 example. The point-line incidence matrix of the generalized quadrangle $W(2)$ is denoted by $B_{15}$ and gives a girth-8 example with larger corank. A connected $30\times 30$ edge-switched enlargement of two $B_{15}$ copies is denoted by $B_{30}$ and gives a larger-corank girth-8 example. A separately selected $17\times 17$ low-corank matrix with the same target weight and the same 4-cycle/6-cycle exclusion conditions is denoted by $B_{17}$. A randomized local-search example with size $18$ is included as a supplementary example and is denoted by $B_{18}$. For column weight four, the incidence matrices of the projective plane $\mathrm{PG}(2,3)$ and the symplectic generalized quadrangle $W(3)$ are denoted by $B_{13}$ and $B_{40}$, respectively~\cite{Hirschfeld1998ProjectiveGeometries,PayneThas2009GQ}. Parameters of the base HGP code obtained from a base matrix $B_s$ carry the subscript $s$. For example, $n_7$ and $k_7$ denote the length and number of logical qubits of the base code obtained from $B_7$, while $n_{15}$ and $k_{15}$ denote the corresponding quantities for $B_{15}$. Quantities after a CPM lift carry both a hat and, when needed, a subscript indicating the base matrix and the lift size, as in $\hat{n}_{15,64}$.

A finite generalized quadrangle is a point-line incidence structure whose incidence graph has diameter four and girth eight; equivalently, it has no geometric triangles and satisfies the usual quadrangle incidence axiom~\cite{PayneThas2009GQ}. The symplectic generalized quadrangle $W(q)$ is the standard finite-geometric example of order $(q,q)$, so it has $(q^2+1)(q+1)$ points and the same number of lines, with $q+1$ points on each line and $q+1$ lines through each point~\cite{PayneThas2009GQ}. Thus $W(2)$ gives the $15\times 15$ column-weight-three incidence matrix used for $B_{15}$, and $W(3)$ gives the $40\times 40$ column-weight-four incidence matrix used for $B_{40}$.

Corollary~\ref{cor:unavoidable-8-cycle-count} gives the number of Tanner 8-cycles that cannot be removed by CPM lifting for these regular square base matrices. On each side, the number of forced 8-cycles is $441$ for the Fano-plane base $B_7$, $2025$ for the generalized-quadrangle base $B_{15}$, $8100$ for the connected enlarged base $B_{30}$, $2601$ for the low-corank base $B_{17}$, $2916$ for the randomized base $B_{18}$, $6084$ for the projective-plane base $B_{13}$, and $57600$ for the generalized-quadrangle base $B_{40}$. In particular, the bases $B_{15}$, $B_{30}$, $B_{17}$, $B_{18}$, and $B_{40}$ have no 4-cycles or 6-cycles after the HGP step, so every CPM $P$ lift from these bases gives Tanner girth $8$ for both lifted matrices $\hat{H}_X$ and $\hat{H}_Z$. Thus increasing $P$ cannot produce girth at least $10$ for these two Tanner graphs as long as the same CPM-lift orthogonality constraints are imposed.

\subsection{Design Criteria for Increasing the Base Number of Logical Qubits}

If the goal is to increase the number of logical qubits in the base code, the direct design target is the corank of the base matrix $B$, not the expanded HGP matrices themselves. For square HGP codes,
\[
\begin{aligned}
n&=2s^2,\qquad k_{\mathrm{des}}=0,\qquad R_{\mathrm{des}}=0,\\
k&=2c_B^2,\qquad R=\frac{k}{n}=\left(\frac{c_B}{s}\right)^2.
\end{aligned}
\]
Thus the design problem in this subsection does not make the design quantum rate positive; it keeps $R_{\mathrm{des}}=0$ and increases the actual rank deficiency. A target value $k_{\mathrm{target}}$ requires a base matrix with
\[
c_B\ge \left\lceil \sqrt{k_{\mathrm{target}}/2}\right\rceil.
\]
Increasing corank alone, however, is not a sufficient finite-length design principle. Low-weight kernel vectors, repeated local structures, and disconnected component decompositions can increase $k$ while degrading distance or decoding behavior.

We use the following design order. First, impose the row and column weight $w$ and the short-cycle conditions as constraints: every row and column has weight $w$, any two rows and any two columns share at most one common 1, and the Tanner graph of $B$ has no simple 6-cycles. Second, maximize the corank under these constraints. This can be specified either by minimizing the rank directly or by choosing a target corank $c_0$ and selecting candidates satisfying $\mathrm{rank}(B)\le s-c_0$. Third, check that the kernel support is not restricted to a small set of coordinates. This means checking for low-weight vectors in $\ker B$, imbalance in row and column supports, and connectedness of the Tanner graph of $B$. Fourth, among candidates with the same $c_B$, prefer those with fewer low-weight classical codewords and fewer low-weight row dependencies. These are finite-length criteria for reducing directly identifiable sources of low-weight logical operators after the HGP step.

There are two structured construction routes for realizing this strategy. One is to use incidence matrices from finite geometries or block designs where regularity, absence of short cycles, and low rank are known from the structure. In this approach, the rank deficiency is explained by symmetry or incidence relations, so the construction is easier to state and verify. The other approach is to prescribe part of the kernel in advance. Choose independent column vectors $\vect{u}_1,\ldots,\vect{u}_{c_0}$ and construct $B$ so that
\[
B\vect{u}_a=\vect{0}\qquad (a=1,\ldots,c_0).
\]
Each row support must then have even overlap with each $\vect{u}_a$. For $w=3$, this means that each row selects three column signatures in $\mathbb{F}_2^{c_0}$ whose sum is zero. After imposing these parity constraints, candidates are filtered by column weight, row-pair and column-pair overlaps, 6-cycles, and connectedness. This guarantees $c_B\ge c_0$ at the design stage, and the final value of $c_B$ is certified by an explicit rank computation.

Disconnected base matrices, or base matrices with a direct-sum structure, can also increase corank. Such growth is usually caused by putting independent components side by side, and it should not be interpreted as an improvement in finite-length distance or decoding behavior. Randomized local search is another possible way to find examples under the same finite checks, but it does not by itself explain why the rank deficiency or short-cycle structure occurs. Therefore, when we say that the base-code value of $k$ is increased, we also require the Tanner graph of $B$ to be connected and the kernel vectors not to be confined to small components.

\subsection{\texorpdfstring{$B_7$}{B7}: Girth-6 Example from the Fano Plane}

The first example does not impose the absence of simple 6-cycles. The Fano plane is the projective plane $\mathrm{PG}(2,2)$: it has seven points and seven lines, each line contains three points, and each point lies on three lines~\cite{Hirschfeld1998ProjectiveGeometries}. Let $B_7$ be its $7\times 7$ point-line incidence matrix. Any two distinct lines meet in one point, and any two distinct points determine one line, so two distinct rows and two distinct columns have at most one common 1. However, the Fano plane contains geometric triangles, which correspond to simple 6-cycles in the point-line incidence graph. Thus $B_7$ gives a girth-6 $(3,6)$-regular CSS-HGP code.

The rank and corank over $\mathbb{F}_2$ are listed in Table~\ref{tab:base-examples-en}. The HGP code obtained from this base matrix is a $(3,6)$-regular CSS-HGP code with parameters $[[98,18,4]]$. Since it contains simple 6-cycles, it is kept separate from the base-matrix designs below that exclude both 4-cycles and 6-cycles.

\subsection{\texorpdfstring{$B_{15}$}{B15}: Finite-Geometric Girth-8 Design}

The next example is a finite-geometric construction. The symplectic generalized quadrangle $W(2)$ is a standard finite-geometric structure with three points on each line and three lines through each point~\cite{PayneThas2009GQ}. Let $B_{15}$ be its $15\times 15$ point-line incidence matrix. This matrix has row and column weight $w=3$. The incidence structure implies that any two rows and any two columns share at most one common 1, and the incidence graph of $W(2)$ has girth $8$, so it has no simple 6-cycles. Thus $B_{15}$ satisfies the design criteria above with $w=3$.

The rank and corank over $\mathbb{F}_2$ are listed in Table~\ref{tab:base-examples-en}. The HGP code obtained from this base matrix is a $(3,6)$-regular CSS code whose $H_X$ and $H_Z$ Tanner graphs both have girth $8$. Here $c_{15}=5$, so the base code has $k_{15}=2c_{15}^2=50$ logical qubits.

\subsection{\texorpdfstring{$B_{30}$}{B30}: Connected Girth-8 Enlargement with Larger Corank}

The next example keeps the same column weight and the same Tanner-girth condition while increasing the base corank. Start from two disjoint copies of $B_{15}$, with the second copy using row and column indices shifted by $15$. This direct sum has size $30$, column weight $3$, Tanner girth $8$, and corank $10$, but its Tanner graph has two connected components. To obtain a connected base matrix, replace the two entries
\[
(0,0),\qquad (15,15)
\]
by the two cross entries
\[
(0,15),\qquad (15,0).
\]
This operation preserves every row weight and column weight. Direct verification gives no row-pair intersections of size at least $2$, no simple 6-cycles, and one connected component. We denote the resulting matrix by $B_{30}$.

The rank and corank over $\mathbb{F}_2$ are listed in Table~\ref{tab:base-examples-en}. The base Tanner graph is connected, and the HGP code obtained from $B_{30}$ is a $(3,6)$-regular CSS-HGP code whose two Tanner graphs have girth $8$. The edge switch reduces the corank from $10$ for the disconnected direct sum to $c_{30}=9$, giving $k_{30}=2c_{30}^2=162$ logical qubits at length $n_{30}=1800$. The HGP distance obtained from the kernel distances of $B_{30}$ and $B_{30}^{\mathsf T}$ is $d_{30}=6$.

\subsection{\texorpdfstring{$B_{17}$}{B17}: Regular Base-Matrix Design with Small Corank}

This example specializes the general $(w,2w)$-regular design to $w=3$ and makes the base value of $k$ small by choosing a base matrix with small positive corank. We represent $B$ as the union of three permutation matrices, which fixes both row and column weights to $3$. Under this constraint, the base matrix is chosen from finite candidates by requiring no row-pair intersections of size at least $2$, no simple 6-cycles, positive corank, and small corank.

Write the selected base matrix as $B_{17}$. It has zero row pairs with intersection $\ge 2$, zero simple 6-cycles, and Tanner girth $8$. Its rank data and the resulting HGP-code parameters are shown in Table~\ref{tab:base-examples-en}. Since $c_{17}=1$, the base code has the smallest positive number of logical qubits in this square-HGP family, namely $k_{17}=2$.

\subsection{\texorpdfstring{$B_{18}$}{B18}: Randomized Local-Search Girth-8 Example}

The preceding examples show known incidence structures and a simple enlargement operation. To record that the same finite checks can also be met by search, we include one supplementary example obtained by randomized local search. The matrix is represented as the union of three permutation matrices, so the row and column weights are fixed to $3$ throughout the search. Local swaps in the permutations are used to explore candidates, and the final matrix is retained only after deterministic verification of the rank, row-pair intersections, simple 6-cycles, Tanner girth, and connectedness. The randomness is therefore used only to find the candidate; the stated properties are direct checks of the displayed matrix.

We denote the resulting $18\times 18$ matrix by $B_{18}$. It has rank $16$, corank $2$, no row-pair intersections of size at least $2$, no simple 6-cycles, Tanner girth $8$, and a connected Tanner graph. The HGP code obtained from this base matrix is a $(3,6)$-regular CSS-HGP code with parameters $[[648,8,12]]$.

\subsection{\texorpdfstring{$B_{13}$}{B13}: A \texorpdfstring{$(4,8)$}{(4,8)}-Regular Girth-6 Example from the Projective Plane \texorpdfstring{$\mathrm{PG}(2,3)$}{PG(2,3)}}

We next include a column-weight-four example obtained from the point-line incidence matrix of the projective plane $\mathrm{PG}(2,3)$. This plane has 13 points and 13 lines; each line contains 4 points, and each point lies on 4 lines~\cite{Hirschfeld1998ProjectiveGeometries}. Let $B_{13}$ denote its $13\times 13$ incidence matrix. Any two distinct lines meet in one point, and any two distinct points determine one line, so two distinct rows and two distinct columns share at most one common 1. On the other hand, projective planes contain geometric triangles, which correspond to simple 6-cycles in the point-line incidence graph, so the Tanner graph of $B_{13}$ has Tanner girth 6.

The rank and corank over $\mathbb{F}_2$, together with the HGP-code parameters, are listed in Table~\ref{tab:base-examples-en}. The HGP code obtained from this base matrix is a $(4,8)$-regular CSS-HGP code with variable column weight 4 and check row weight 8, and its parameters are $[[338,2,13]]$.

\subsection{\texorpdfstring{$B_{40}$}{B40}: A \texorpdfstring{$(4,8)$}{(4,8)}-Regular Girth-8 Example from the Generalized Quadrangle \texorpdfstring{$W(3)$}{W(3)}}

The final example is the $w=4$ finite-geometric construction from the point-line incidence matrix of the symplectic generalized quadrangle $W(3)$. This generalized quadrangle has 40 points and 40 lines; each line contains 4 points, and each point lies on 4 lines~\cite{PayneThas2009GQ}. Let $B_{40}$ denote its $40\times 40$ incidence matrix. The incidence graph of a generalized quadrangle has girth 8, so the Tanner graph of $B_{40}$ has no 4-cycles and no simple 6-cycles. Thus this example satisfies the 4-cycle/6-cycle exclusion conditions at $w=4$.

The rank and corank over $\mathbb{F}_2$, together with the HGP-code parameters, are listed in Table~\ref{tab:base-examples-en}. The HGP code obtained from this base matrix is a $(4,8)$-regular CSS-HGP code with variable column weight 4 and check row weight 8, and its parameters are $[[3200,450,8]]$.

The distance column $d$ is evaluated from the standard HGP distance formula
\[
d=\min(d_B,d_{B^{\mathsf T}}),
\]
where $d_B$ and $d_{B^{\mathsf T}}$ are the kernel distances of $B$ and its transpose~\cite{TillichZemor2014}.

The seven codes in Table~\ref{tab:base-examples-en} are obtained by applying the HGP construction to square binary base matrices. In this structural sense, they all belong to the known HGP construction class. The $[[98,18,4]]$ code obtained from the Fano-plane incidence matrix $B_7$ appears in the literature as a finite-length HGP example~\cite{HuangZhouFangZhaoYing2025VeriQEC}. The projective-plane and generalized-quadrangle incidence matrices $B_{13}$, $B_{15}$, and $B_{40}$ give explicit finite-geometric instances, $B_{30}$ is a connected edge-switched enlargement used to increase corank, $B_{17}$ is a low-corank instance selected according to the base-matrix design criteria, and $B_{18}$ is a randomized local-search instance. The table gives common columns for regular degree, Tanner girth, rank/corank, and finite-length parameters.

In Table~\ref{tab:base-examples-en}, $s$ is the base-matrix size, $\rho_B$ and $c_B$ are the rank and corank of the base matrix over $\mathbb F_2$, $d$ is the HGP distance obtained from the kernel distances above, $m_X,m_Z$ are check-row counts, $\rho_X,\rho_Z$ are check-matrix ranks, and $(d_{\mathrm v},d_{\mathrm c})$ gives the variable column weight and check row weight of the HGP Tanner graph.

\begin{table}[t]
\centering
\caption{Parameters of the regular base CSS-HGP codes used in Section 3}
\label{tab:base-examples-en}
\resizebox{\textwidth}{!}{%
\begin{tabular}{llcccccccccc}
\toprule
base matrix & construction & $s$ & $\rho_B$ & $c_B$ & Tanner girth of $B$ & $n$ & $k$ & $d$ & $m_X=m_Z$ & $\rho_X=\rho_Z$ & $(d_{\mathrm v},d_{\mathrm c})$ \\
\midrule
$B_7$ & Fano plane & $7$ & $4$ & $3$ & $6$ & $98$ & $18$ & $4$ & $49$ & $40$ & $(3,6)$ \\
$B_{15}$ & $W(2)$ generalized quadrangle & $15$ & $10$ & $5$ & $8$ & $450$ & $50$ & $6$ & $225$ & $200$ & $(3,6)$ \\
$B_{30}$ & connected edge-switched enlargement & $30$ & $21$ & $9$ & $8$ & $1800$ & $162$ & $6$ & $900$ & $819$ & $(3,6)$ \\
$B_{17}$ & low-corank design & $17$ & $16$ & $1$ & $8$ & $578$ & $2$ & $8$ & $289$ & $288$ & $(3,6)$ \\
$B_{18}$ & randomized local search & $18$ & $16$ & $2$ & $8$ & $648$ & $8$ & $12$ & $324$ & $320$ & $(3,6)$ \\
$B_{13}$ & $\mathrm{PG}(2,3)$ & $13$ & $12$ & $1$ & $6$ & $338$ & $2$ & $13$ & $169$ & $168$ & $(4,8)$ \\
$B_{40}$ & $W(3)$ generalized quadrangle & $40$ & $25$ & $15$ & $8$ & $3200$ & $450$ & $8$ & $1600$ & $1375$ & $(4,8)$ \\
\bottomrule
\end{tabular}%
}
\end{table}

\section{\texorpdfstring{CPM $P$ Lift}{CPM P Lift}}

The CPM lift is described by assigning a shift in $\mathbb Z_P$ to each nonzero block of the base CSS matrices. If the nonzero entry in row $i$ and column $c$ of $H_X$ is assigned shift $s^X_{i,c}$, and the corresponding entry of $H_Z$ is assigned shift $s^Z_{j,c}$, then the lifted entries are $P\times P$ circulant permutation matrices with those shifts.

CSS orthogonality is imposed as a system of linear congruences. For every pair of rows $i$ of $H_X$ and $j$ of $H_Z$, let $C_{ij}$ be the set of columns shared by the two rows. Since the base CSS matrices are orthogonal, $|C_{ij}|$ is even. We pair the shared columns and impose, for each pair $c,c'\in C_{ij}$,
\[
s^X_{i,c}-s^Z_{j,c}
-
s^X_{i,c'}+s^Z_{j,c'}
\equiv 0 \pmod P.
\]
These equations ensure that the lifted overlap between any $X$-check and any $Z$-check is even in every sheet difference, and therefore the lifted matrices still satisfy $\hat{H}_X\hat{H}_Z^{\mathsf T}=0$.

Short-cycle suppression is expressed by nonzero voltage conditions. For a Tanner cycle $C$ in the base graph, define its signed voltage sum $V(C)$ by adding the shifts along the cycle with alternating signs according to the direction of traversal. The lift contains a cycle above $C$ with the same length only when
\[
V(C)\equiv 0 \pmod P.
\]
Thus the construction problem is a finite constraint-satisfaction problem of the form
\[
A\vect{s}\equiv \vect{0}\pmod P,
\qquad
\vect{v}_C^{\mathsf T}\vect{s}\not\equiv 0\pmod P
\]
for the targeted short cycles $C$. Here $\vect{s}$ is the column vector collecting all shift variables, $A$ is the integer coefficient matrix of the CSS-orthogonality congruences, and $\vect{v}_C$ is the signed coefficient vector whose inner product with $\vect{s}$ equals the cycle voltage $V(C)$. The first part is linear, while the second part is a family of nonzero constraints. We use a randomized feasible walk in the kernel of the linear system: starting from a feasible lift, we propose random kernel moves, and accept only moves that keep every targeted nonzero cycle-voltage constraint satisfied. This procedure is not a uniform sampler over all feasible lifts; it is used only to obtain a shift assignment satisfying the stated zero and nonzero congruence constraints.

\subsection{\texorpdfstring{Randomized $P=64$ Lift of the Generalized-Quadrangle Base $B_{15}$}{Randomized P=64 Lift of the Generalized-Quadrangle Base B15}}

For the generalized-quadrangle base matrix $B_{15}$, we also construct a randomized $P=64$ lift with the same degree pair $(3,6)$. In this case, some Tanner 8-cycles are forced by the CSS orthogonality equations. We therefore remove only those forced 8-cycles from the nonzero-constraint list and require all remaining targeted cycles of length at most 10 to have nonzero voltage. The resulting lifted code has the verified parameters
\[
\hat{n}_{15,64}=28800,\qquad
\hat{m}_{X,15,64}=\hat{m}_{Z,15,64}=14400,
\]
\[
\hat{\rho}_{X,15,64}=\hat{\rho}_{Z,15,64}=14369,
\qquad
\hat{k}_{15,64}=62.
\]
The row weight remains $6$ and the column weight remains $3$. Direct orthogonality verification gives zero odd-overlap pairs between $\hat{H}_X$ and $\hat{H}_Z$. Direct graph search on the actual lifted matrices also shows that the Tanner graph of $\hat{H}_X$ and the Tanner graph of $\hat{H}_Z$ are both connected; each graph has one connected component containing all $14400$ check vertices and all $28800$ variable vertices. Thus the increase from the base value $k_{15}=50$ to $\hat{k}_{15,64}=62$ is not explained by a decomposition into disconnected lifted copies. The short-cycle constraints used in the construction include all targeted avoidable 8-cycle constraints and all targeted 10-cycle constraints. The forced 8-cycles remain, so the Tanner graphs of $\hat{H}_X$ and $\hat{H}_Z$ each have girth $8$. This girth statement refers to the Tanner graph of each parity-check matrix separately; the combined graph containing both $X$- and $Z$-checks can contain 4-cycles from even $X/Z$ overlaps required by CSS orthogonality. The shift tables and verification scripts for this lift are provided as supplementary material.

We also give explicit distance upper-bound certificates. Index a lifted qubit by $(v,t)$, where $v\in\{0,\ldots,449\}$ is the base-column index and $t\in\mathbb Z_{64}$ is the fiber coordinate, and let $E=\{0,2,\ldots,62\}$. An $X$-type representative is supported on
\[
\{(v,t):v\in\{308,323,338,353,368,383\},\ t\in E\},
\]
and a $Z$-type representative is supported on
\[
\{(v,t):v\in\{4,19,49,64,109,139\},\ t\in E\}.
\]
Both vectors have weight $6\cdot32=192$. Direct binary computation gives $\hat H_Z\vect{x}^{\mathsf T}=0$ and $\hat H_X\vect{z}^{\mathsf T}=0$. Row reduction gives
\[
\begin{aligned}
\operatorname{rank}\hat H_X&=14369,
&\operatorname{rank}\begin{bmatrix}\hat H_X\\ \vect{x}\end{bmatrix}&=14370,\\
\operatorname{rank}\hat H_Z&=14369,
&\operatorname{rank}\begin{bmatrix}\hat H_Z\\ \vect{z}\end{bmatrix}&=14370.
\end{aligned}
\]
Thus neither representative is a stabilizer and
\[
d_X\le192,\qquad d_Z\le192,\qquad d\le192.
\]
The support lists and coordinate convention are also included in the supplementary material.
Complete enumeration on the same fixed matrices excludes every nontrivial logical support of even weight at most $16$ on both sides, while allowing stabilizer supports. Since every column has weight three, kernel vectors have even weight. The enumeration and its completeness argument are given in Appendix~\ref{sec:distance-enumeration}. Thus
\[
18\le d_X\le192,\qquad 18\le d_Z\le192,\qquad 18\le d\le192.
\]

\subsection{Connectedness and Number of Logical Qubits}

When interpreting the number of logical qubits after lifting, it is necessary to separate trivial growth caused by disconnected lifts from growth inside a connected lift. For example, if every nonzero entry is replaced by the same $P\times P$ identity permutation, then the lifted Tanner graph is the disjoint union of $P$ copies of the base Tanner graph. Up to row and column permutations, the lifted matrices are direct sums, and
\[
\hat{k} = Pk.
\]
More generally, if a disconnected lift decomposes into several connected components, then the total value of $k$ is the sum of the values of $k$ for those components. Thus, allowing disconnectedness makes it easy to increase the absolute value of $k$. This increase, however, only duplicates code components: the rate is unchanged, and the minimum distance is at most the minimum of the component distances. The nontrivial question in this paper is therefore whether one can obtain additional logical qubits while keeping the lift connected.

A CPM lift can be viewed as a graph cover described by a voltage graph~\cite{GrossTucker1987TopologicalGraphTheory,Fossorier2004CirculantPermutation}. When the base Tanner graph is connected and the lift uses the cyclic group $\mathbb Z_P$, the standard connectedness condition is that the subgroup of $\mathbb Z_P$ generated by voltage sums along closed walks is the whole group $\mathbb Z_P$. Equivalently, for cyclic lifts, the greatest common divisor of $P$ and the generated voltage sums must be $1$. For a general permutation lift, the corresponding condition is that the monodromy group generated by closed walks acts transitively on the $P$ sheets.

Increasing $k$ while preserving connectedness requires additional row dependencies that are not inherited from the base code. Let $\hat{\ell}_X=\hat{m}_X-\hat{\rho}_X$ and $\hat{\ell}_Z=\hat{m}_Z-\hat{\rho}_Z$ be the lifted left-kernel dimensions. In the square HGP setting of this paper, $\hat{n}=\hat{m}_X+\hat{m}_Z$, so
\[
\hat{k}=\hat{\ell}_X+\hat{\ell}_Z.
\]
Theorem~\ref{thm:cpm-lift-k-lower-bound} accounts only for base row dependencies lifted as vectors that are constant on the $P$ sheets. Therefore, obtaining $\hat{k}>k$ in a connected lift requires additional left-kernel vectors that are nonconstant in the sheet direction, namely vectors satisfying
\[
\hat{H}_X^{\mathsf T}\hat{\vect{u}}=\vect{0}
\quad\text{or}\quad
\hat{H}_Z^{\mathsf T}\hat{\vect{v}}=\vect{0},
\]
without destroying connectedness or the targeted short-cycle constraints.

This suggests the following construction strategies. First, impose CSS orthogonality, the short-cycle constraints, and connectedness as hard constraints, and then use rank as a secondary objective by minimizing $\hat{\rho}_X+\hat{\rho}_Z$, equivalently maximizing $\hat{\ell}_X+\hat{\ell}_Z$. Second, use a disconnected high-$k$ lift as an initial point and perturb shifts so that the generated voltage subgroup becomes all of $\mathbb Z_P$, accepting only perturbations that preserve rank deficiency. Third, prescribe candidate nonconstant left-kernel vectors and add the equations they impose to the lift constraints. None of these strategies is guaranteed to succeed: extra rank deficiency can conflict with short-cycle suppression and often disappears when the lift is made connected. Therefore, any claim of connected growth of $k$ must be certified by an explicit rank computation for the final lifted matrices.

\section{Decoding Method}

Decoding is performed by the BP decoder used in~\cite{OkadaKasai2026AffineCoset} for binary CSS codes under depolarizing noise. Ordered-statistics decoding (OSD) originates as a reliability-based classical decoding method~\cite{FossorierLin1995OSD}; in quantum LDPC decoding, BP followed by OSD-type repair has been used as an effective finite-length strategy~\cite{PoulinChung2008Iterative,PanteleevKalachev2021Degenerate}. Here ``binary'' means that the parity checks are binary sparse matrices $H_X$ and $H_Z$, while each variable node still carries a four-state Pauli prior over $I/X/Z/Y$. Messages exchanged with checks are binary $x$-bit and $z$-bit messages.

Rows of $H_X$ constrain the $z$ component, and rows of $H_Z$ constrain the $x$ component. The decoder supports two notions of success:
\begin{itemize}
\item exact success: the estimated error equals the sampled true error,
\item degenerate success: the residual error differs from the true error only by stabilizers, i.e., it lies in the corresponding row spaces of $H_X$ and $H_Z$.
\end{itemize}

The simulations use a degeneracy-aware success criterion. After BP, the reported decoder applies ordered-statistics-decoding-lite (OSD-lite) post-processing. In this paper, OSD-lite denotes a restricted reliability-based repair step applied only when BP terminates with a nonzero residual syndrome; it is not a full high-order OSD decoder. For each unrepaired side, candidate columns are ordered by a BP-derived flip cost computed from the four-state Pauli posterior, with smaller cost treated as less reliable and therefore more plausible to flip. The main repair step uses a joint prefix-search rule: the decoder solves the induced binary syndrome equation over $\mathbb F_2$ on reliability-ordered prefixes, applies a unique solution only when its support weight is at most $30$, and otherwise uses a restricted nonunique fallback only when the residual degree of freedom is at most $5$ and the selected correction weight is at most $20$. As a local fallback, candidate columns are selected from variables incident to unsatisfied checks and, when the number of unsatisfied checks is at most $32$, from their two-hop neighborhood. The local candidate-pool parameter is fixed to $2048$ variables in the plotted runs, with effective caps of $160$, $128$, and $96$ when the number of unsatisfied checks is larger than $6$, $12$, and $24$, respectively. The joint exact repair stage is capped at $128$ variables and $20$ degrees of freedom, and the ETS lookup table is disabled. These repair limits are fixed for all plotted points.

\section{Numerical Results}

\subsection{Experimental Setup}

The construction section defines several base matrices and lift choices, but the numerical experiments below use lifts derived from the generalized-quadrangle base $B_{15}$, especially the randomized $P=64$ lift. Performance is reported as the frame error rate (FER). For a point with $N$ independent decoding trials and $f$ failed frames, the plotted estimate is $\mathrm{FER}=f/N$, and the shaded band is the Wilson binomial 95\% confidence interval for this failure probability. When $f=0$, the marker is placed at the upper endpoint of the Wilson interval because $\mathrm{FER}=0$ cannot be displayed on a logarithmic vertical axis.

For comparison, the plot shows the depolarizing-channel hashing limit as the root $p_{\mathrm{hash}}\approx0.18929$ of $1-H_2(p)-p\log_2 3=0$, where $H_2$ is the binary entropy function~\cite{BennettDiVincenzoSmolinWootters1996}. The plot also includes the population-dynamics estimate $p_{\mathrm{DE}}\approx0.1529$ of the BP density-evolution threshold for the cycle-free $(3,6)$-regular CSS-LDPC ensemble on the depolarizing channel~\cite{RichardsonUrbanke2001Capacity,PoulinChung2008Iterative}. This value is an asymptotic CSS-LDPC comparison value for the same local degrees; it is not a direct threshold claim for any individual finite lifted code.

\begin{figure}[t]
\centering
\includegraphics[width=0.92\linewidth]{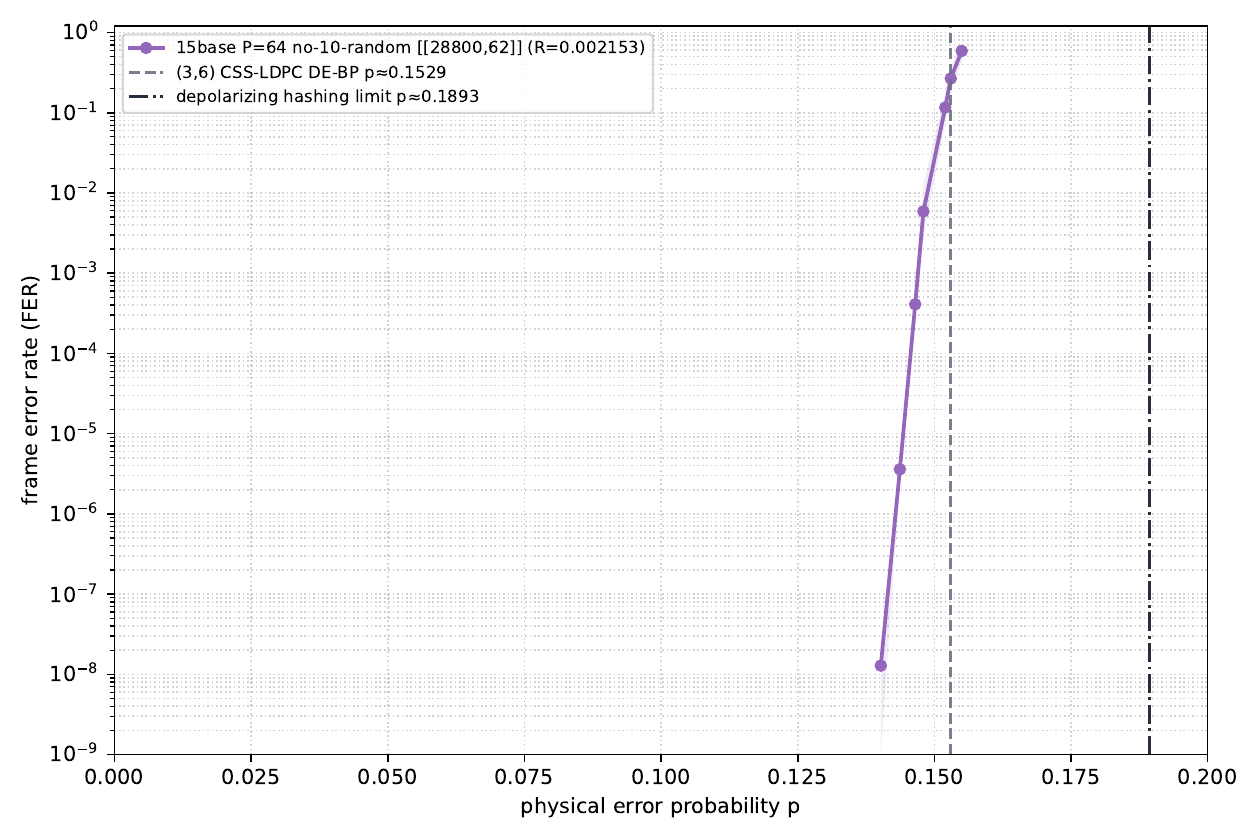}
\caption{$p$--FER plot for the randomized $P=64$ lift derived from the generalized-quadrangle base $B_{15}$. Each point aggregates independent decoding trials at the same $p$, using the degeneracy-aware BP decoder followed by OSD-lite post-processing; shaded bands indicate Wilson 95\% confidence intervals. Zero-failure points are plotted at the upper endpoint of the Wilson interval. The dashed vertical line is the population-dynamics estimate $p_{\mathrm{DE}}\approx0.1529$ of the BP density-evolution threshold for the $(3,6)$-regular CSS-LDPC ensemble on the depolarizing channel, and the dash-dotted vertical line is the depolarizing-channel hashing limit $p_{\mathrm{hash}}\approx0.18929$.}
\label{fig:p-vs-fer-en}
\end{figure}

Figure \ref{fig:p-vs-fer-en} gives the measured $p$--FER values for the randomized $P=64$ lift based on the generalized-quadrangle base $B_{15}$. Other base matrices and lifts are included in the construction section, but they are not decoded in separate experimental campaigns in this paper. The point at $p=0.1402$ aggregates $2.993\times 10^8$ independent decoding trials with zero observed failures, so its marker is placed at the Wilson 95\% upper endpoint $1.28\times 10^{-8}$. The average number of BP iterations increases in the transition region, which indicates that this code and decoder require more iterations there.

\subsection{Degenerate Success and Repair Behavior}

In the recorded successful trials, the dominant success mode is degenerate success rather than exact success. In other words, the decoder usually returns a stabilizer-equivalent correction rather than reproducing the exact sampled Pauli error. This is expected for stabilizer codes and confirms that degeneracy is necessary for interpreting the reported success rates.

All recorded decoding failures in these runs are syndrome-mismatch failures: after the decoder terminates, the residual syndrome is nonzero. Consequently, these failures are not syndrome-consistent logical operators and do not provide a distance upper bound. The decoding data are therefore used here as performance evidence, independently of the explicit structural certificate $d\le192$ above. OSD-lite provides additional repair successes at $p=0.143$, $0.145$, $0.150$, and $0.158$, but the current data do not indicate a substantial change in the FER transition due to that repair step.

\section{Discussion}

The construction results separate degree regularity, rank deficiency, and short-cycle structure in a square-base HGP family. The same design criteria give $(w,2w)$-regular CSS-HGP codes, with the explicit examples covering both $(3,6)$ and $(4,8)$ degrees. Since the design quantum rate is zero, all positive rate comes from corank in the base matrix and additional rank deficiency after lifting. CPM lifting preserves orthogonality but cannot remove the orthogonality-forced 8-cycles identified here; the decoding experiments evaluate the selected column-weight-three generalized-quadrangle lift under a degeneracy-aware BP+OSD-lite decoder.

The current limitations are as follows. The certified intervals are $18\le d_X,d_Z,d\le192$; tightening them and determining the exact distance remain open. The recorded decoding failures have nonzero residual syndrome, so they do not improve these structural upper bounds. Some high-noise FER points have large statistical uncertainty because their sample counts are small. Systematic comparisons across multiple lifts or multiple decoder variants remain incomplete.

\section{Conclusion}

This paper studied square-base CSS-HGP constructions using explicit $(3,6)$ and $(4,8)$ examples. The examples are organized mainly around known finite-geometric incidence structures and simple combinatorial modifications, while a randomized local-search instance is included only to show that the same finite checks can also be met by search. The $(3,6)$ examples are the Fano-plane girth-6 construction, the generalized-quadrangle girth-8 construction, a connected edge-switched girth-8 enlargement with larger corank, a low-corank girth-8 construction, and a randomized local-search girth-8 construction, while the $(4,8)$ examples come from the projective plane $\mathrm{PG}(2,3)$ and the generalized quadrangle $W(3)$. All examples have design quantum rate zero, and their positive values of $k$ are rank deficiencies. We used the standard voltage-sum criterion to identify CSS-orthogonality-forced Tanner 8-cycles, and specialized this statement to the square-HGP family to count forced 8-cycles. For the randomized $P=64$ lift of $B_{15}$, complete support enumeration and explicit non-stabilizer representatives give $18\le d_X,d_Z,d\le192$. The reported CPM-lift decoding experiments are limited to column-weight-three generalized-quadrangle-based lifts, especially this instance. Future work should add exact or tighter distance certification, systematic repair comparisons, replicated statistics across several lifted instances, and lift-design methods that increase the actual number of logical qubits $\hat{k}$ while preserving connectedness, CSS orthogonality, and the targeted short-cycle constraints. In particular, this requires constructing connected lifts with additional nonconstant left-kernel vectors, rather than relying only on the corank inherited from the base matrix.

\appendix

\section{Visual Representations of Base and Check Matrices}

This appendix gives visual representations of the square base matrices used in Table~\ref{tab:base-examples-en}. A black square denotes an entry equal to $1$ over $\mathbb F_2$, and a white square denotes an entry equal to $0$. The row and column orderings are the orderings used in the constructions in Section~3.

\begin{figure}[p]
\centering
\begin{minipage}{0.18\linewidth}
\centering
\includegraphics[width=\linewidth]{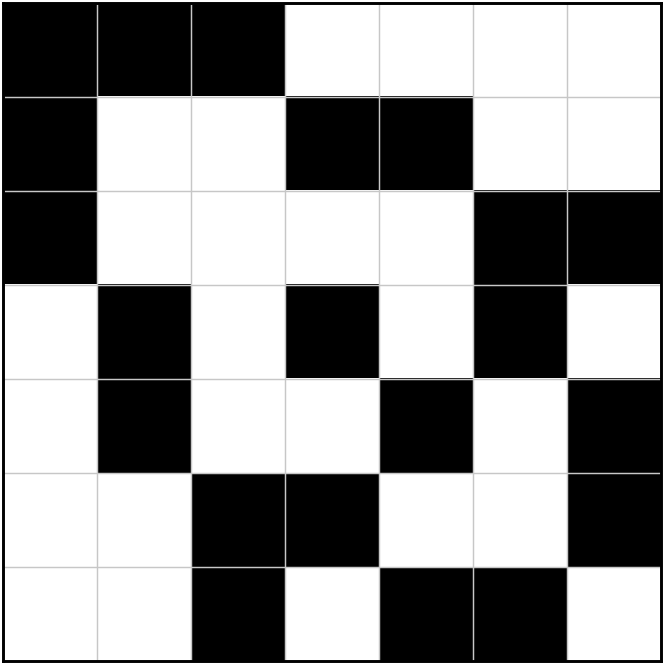}
\par\vspace{0.2em}\small $B_7$
\end{minipage}\hfill
\begin{minipage}{0.18\linewidth}
\centering
\includegraphics[width=\linewidth]{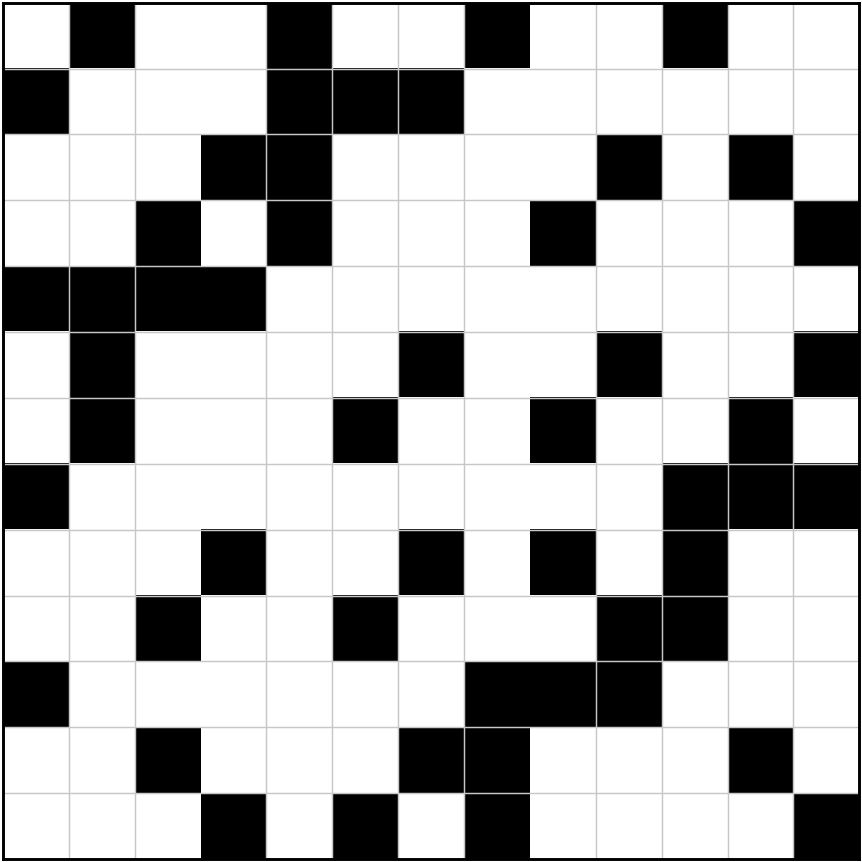}
\par\vspace{0.2em}\small $B_{13}$
\end{minipage}\hfill
\begin{minipage}{0.18\linewidth}
\centering
\includegraphics[width=\linewidth]{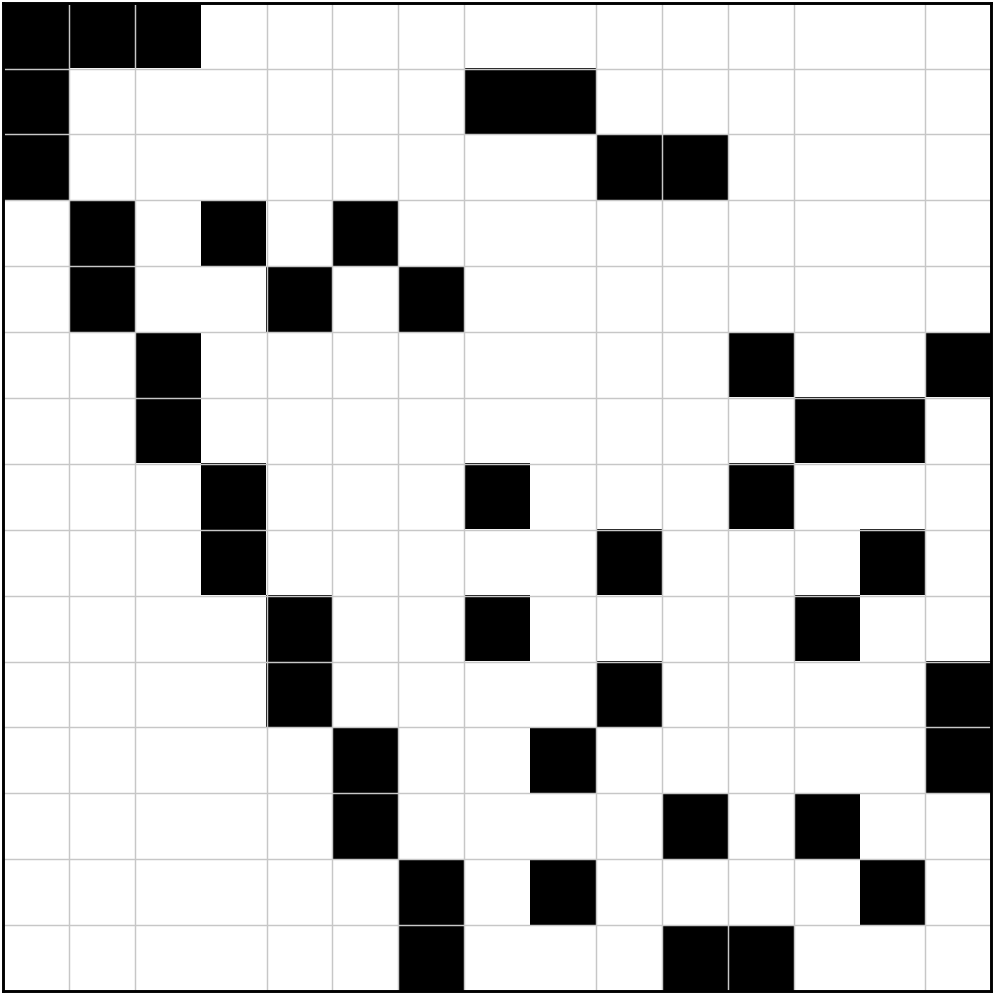}
\par\vspace{0.2em}\small $B_{15}$
\end{minipage}\hfill
\begin{minipage}{0.18\linewidth}
\centering
\includegraphics[width=\linewidth]{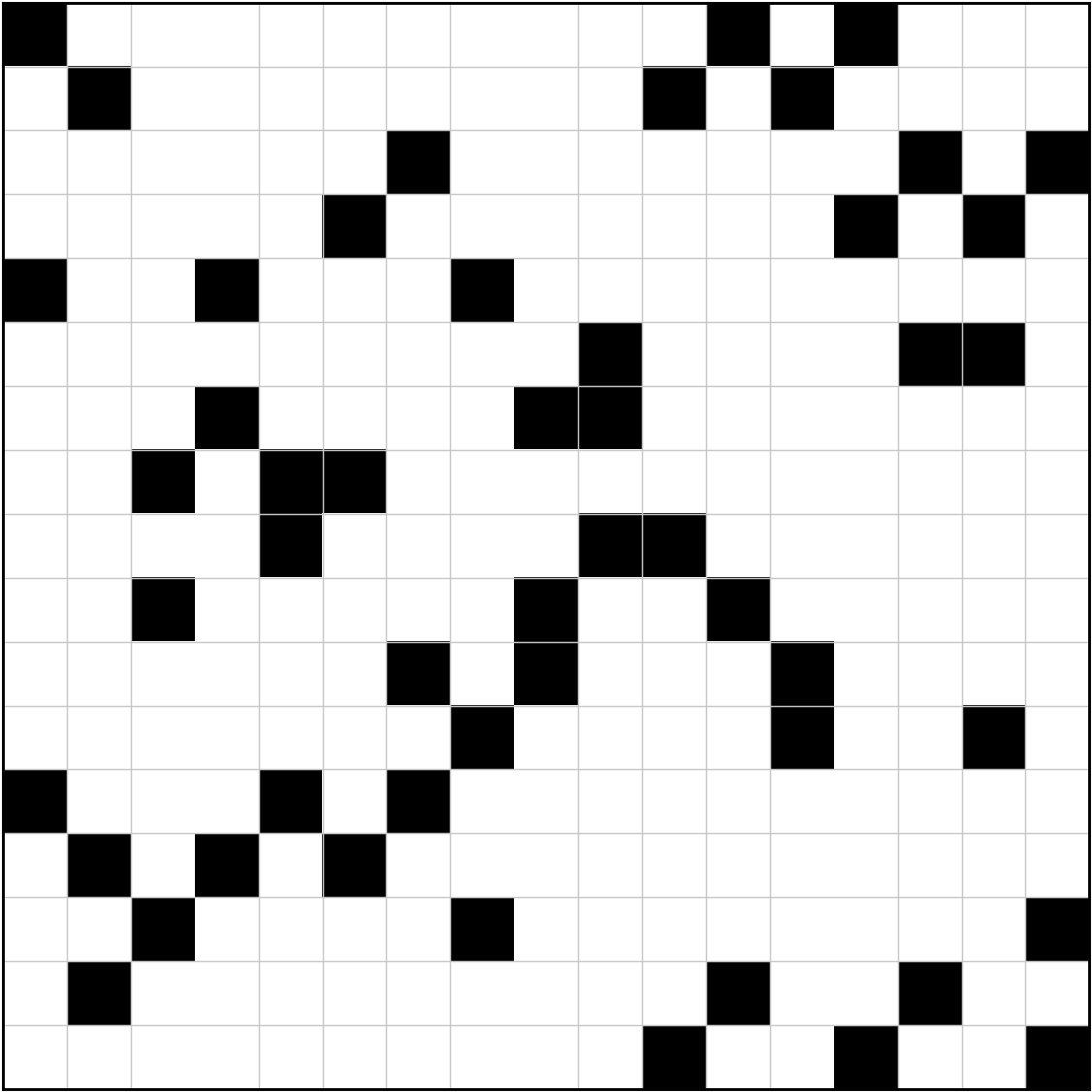}
\par\vspace{0.2em}\small $B_{17}$
\end{minipage}\hfill
\begin{minipage}{0.18\linewidth}
\centering
\includegraphics[width=\linewidth]{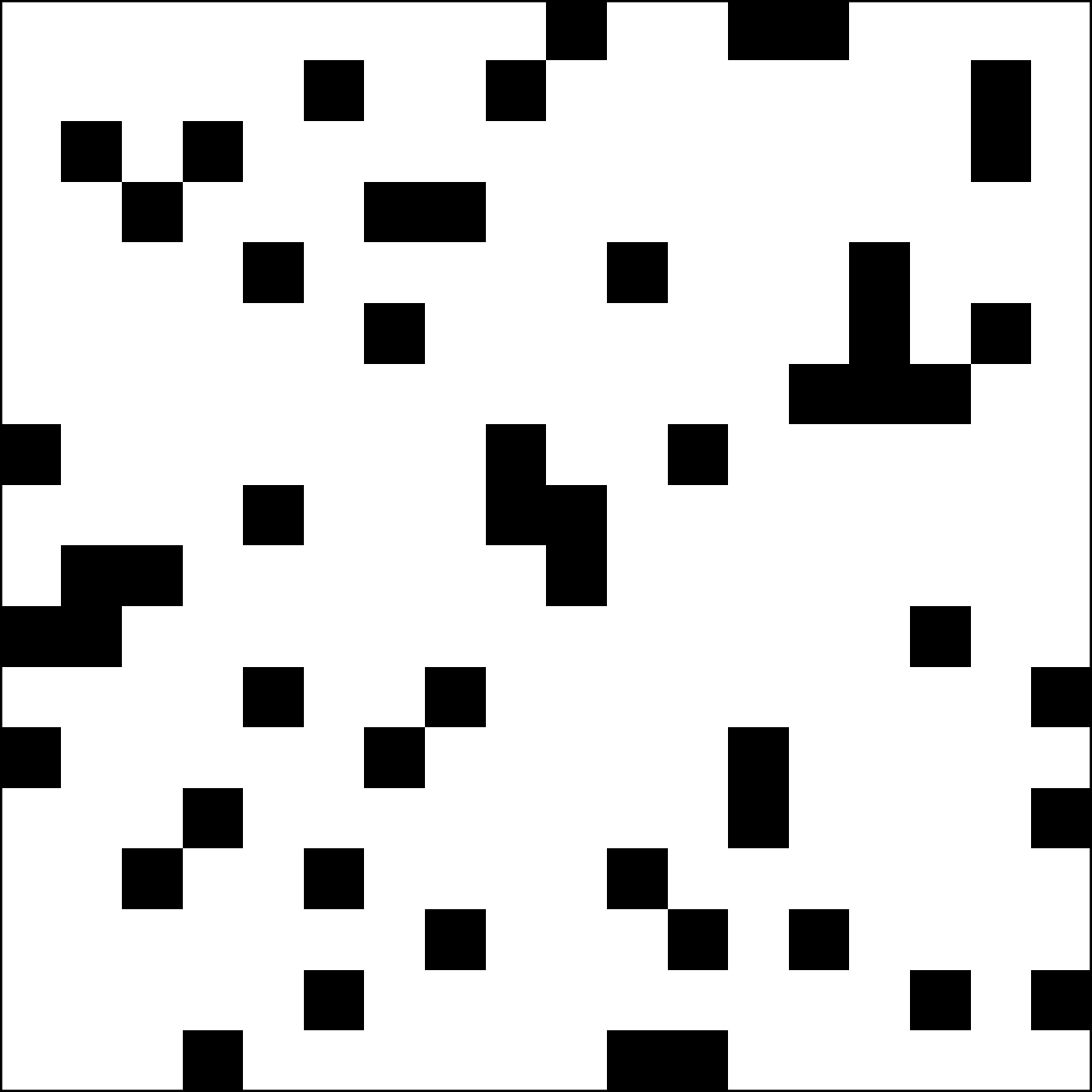}
\par\vspace{0.2em}\small $B_{18}$
\end{minipage}
\caption{Matrix plots of the smaller square base matrices.}
\label{fig:base-matrix-plots-small-en}
\end{figure}

\begin{figure}[p]
\centering
\begin{minipage}{0.44\linewidth}
\centering
\includegraphics[width=\linewidth]{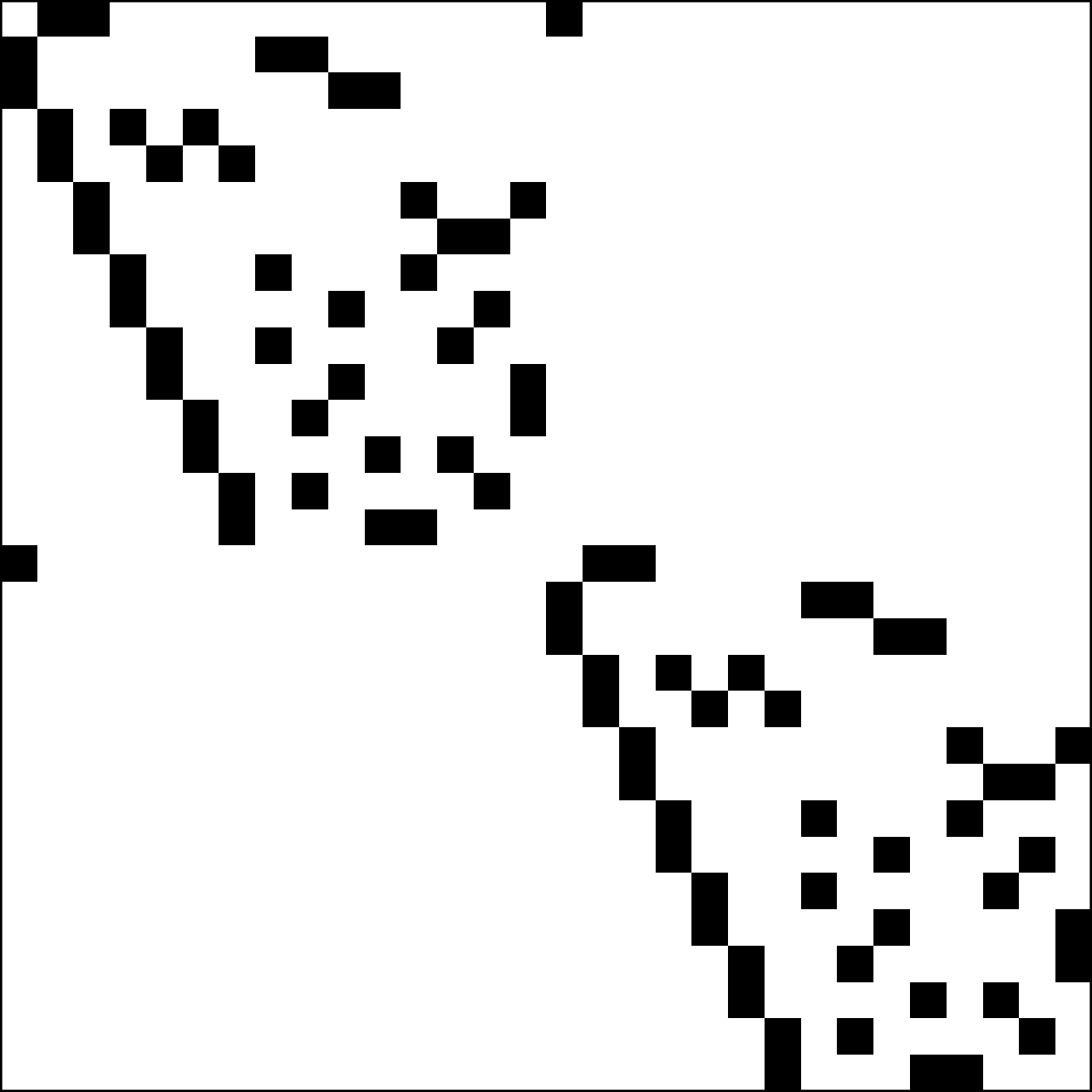}
\par\vspace{0.2em}\small $B_{30}$
\end{minipage}\hfill
\begin{minipage}{0.44\linewidth}
\centering
\includegraphics[width=\linewidth]{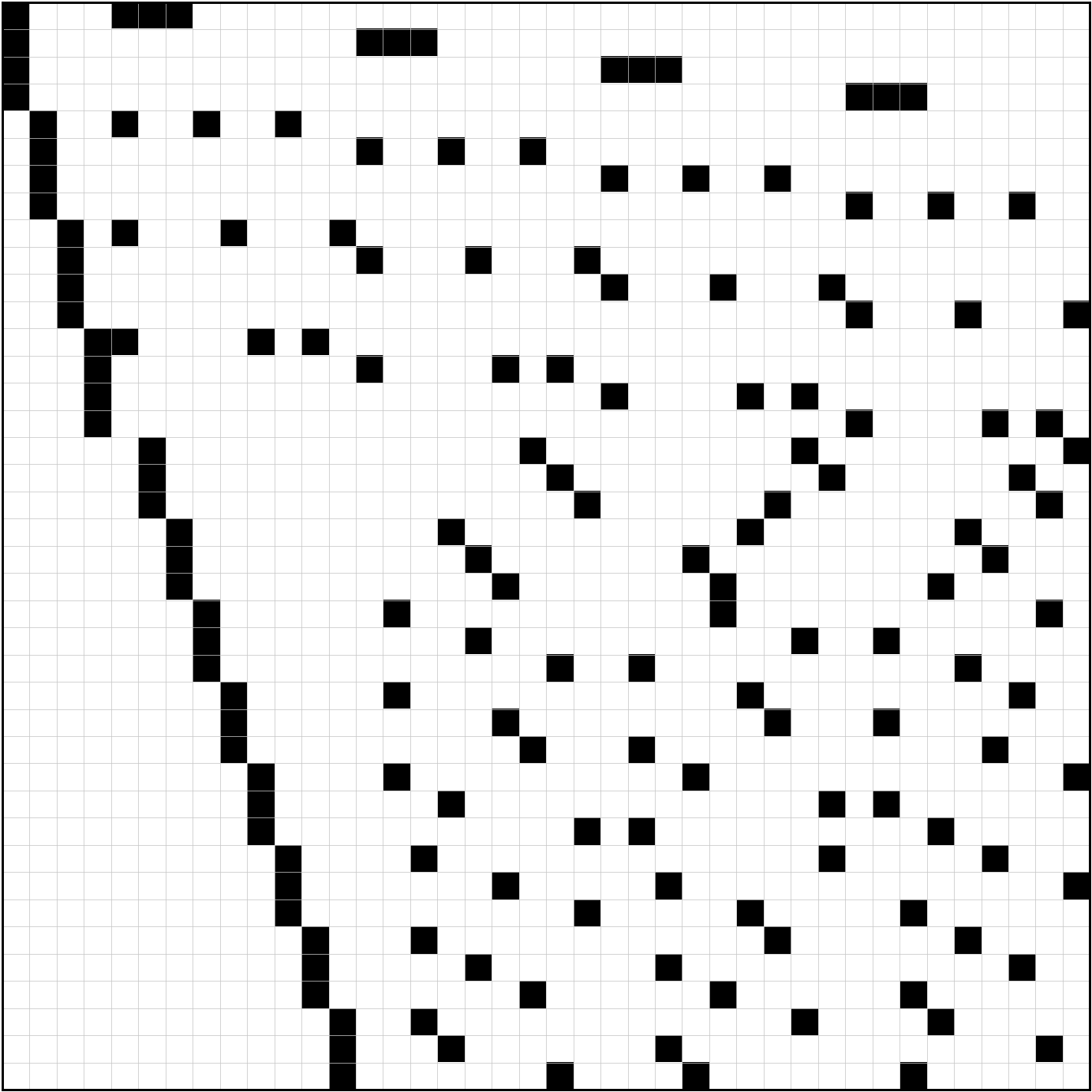}
\par\vspace{0.2em}\small $B_{40}$
\end{minipage}
\caption{Matrix plots of the larger square base matrices.}
\label{fig:base-matrix-plots-large-en}
\end{figure}

As a concrete example of the HGP check matrices, Fig.~\ref{fig:b7-hgp-check-matrix-plots-en} shows the matrices obtained from the Fano-plane base matrix $B_7$. Both matrices have size $49\times 98$. The vertical line separates the first $s^2$ variable columns and the second $s^2$ variable columns in the block definitions of $H_X(B)$ and $H_Z(B)$.

\begin{figure}[p]
\centering
\includegraphics[width=\linewidth]{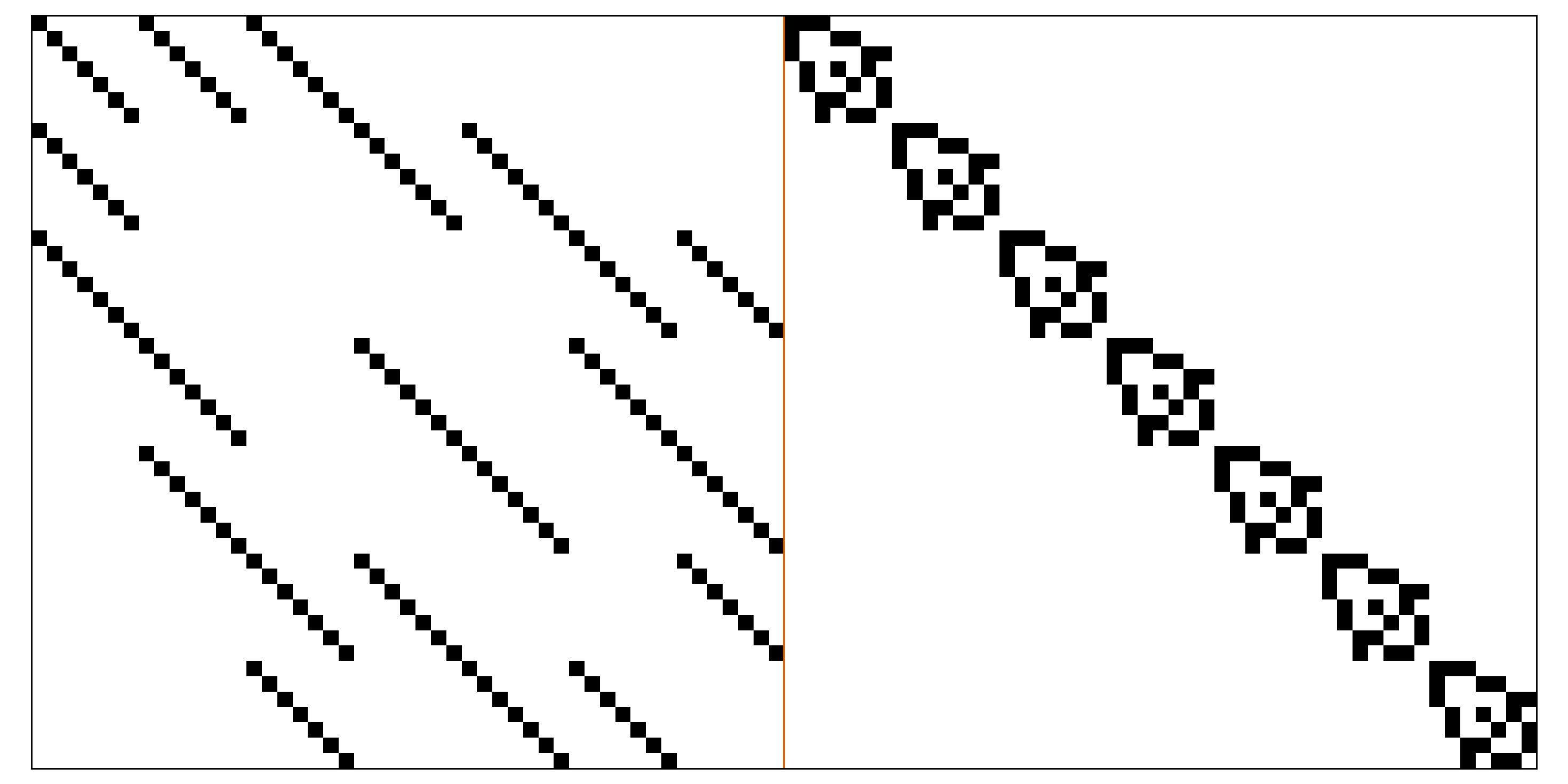}
\par\vspace{0.3em}\small $H_X(B_7)$
\vspace{1.0em}

\includegraphics[width=\linewidth]{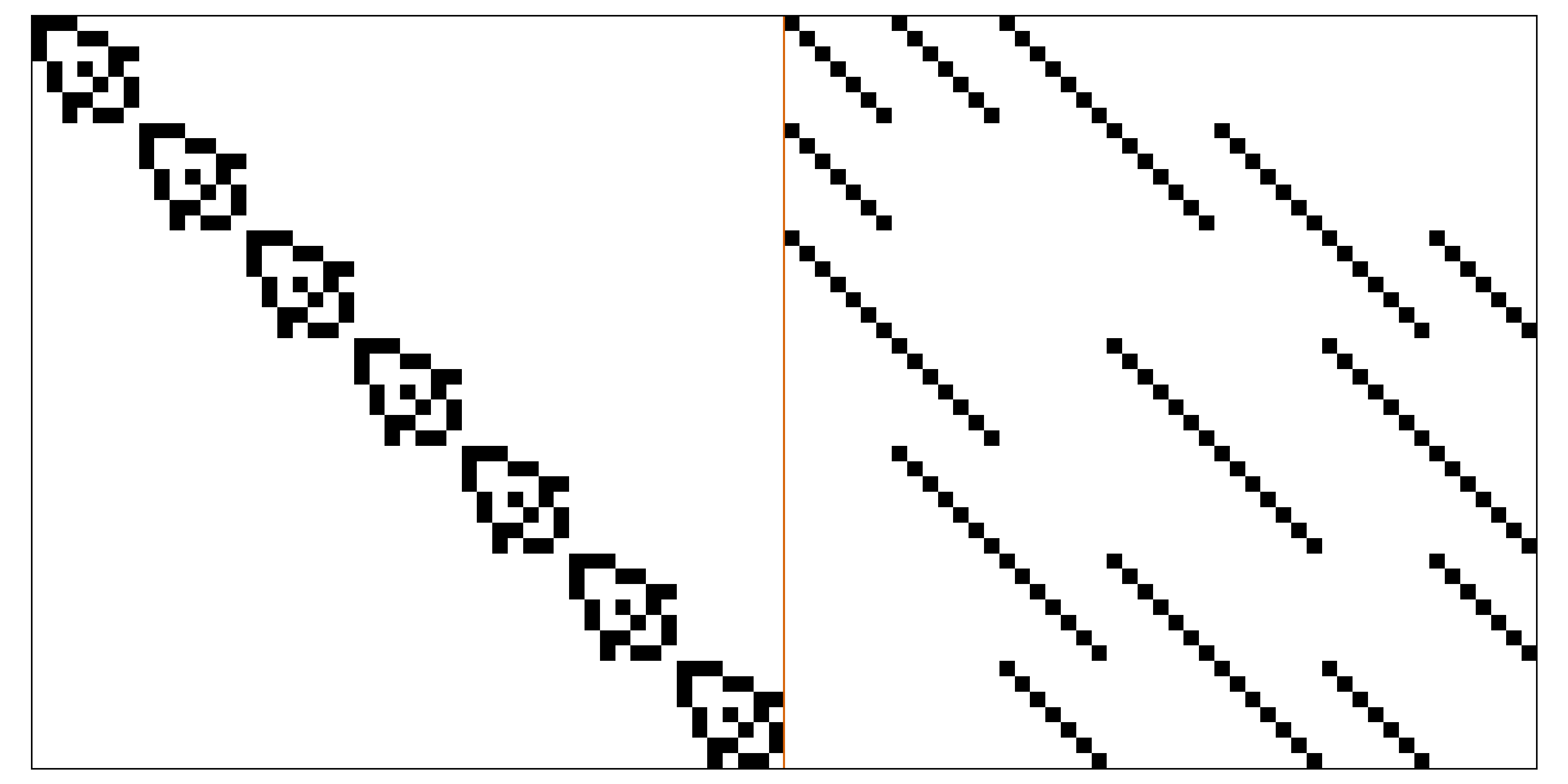}
\par\vspace{0.3em}\small $H_Z(B_7)$
\caption{Matrix plots of the HGP check matrices obtained from the Fano-plane base matrix $B_7$.}
\label{fig:b7-hgp-check-matrix-plots-en}
\end{figure}

\clearpage

\input{distance_enumeration}

\bibliographystyle{IEEEtran}
\bibliography{qldpc_refs}

\end{document}

%% file: distance_enumeration.tex
\section{Computer-Assisted Distance Lower Bounds}\label{sec:distance-enumeration}
For a binary CSS pair, the distances are the minimum weights in
$\ker H_Z\setminus\operatorname{row}(H_X)$ and
$\ker H_X\setminus\operatorname{row}(H_Z)$, respectively.
The computation tests membership in the stabilizer row space by exact binary elimination.
Light stabilizers are retained in the code and are not mistaken for logical operators.
Since every column of each matrix has odd weight three, summing all check equations shows that every kernel vector has even weight.

A minimum-weight nontrivial logical support can be taken to be connected under any pairing of its incident edges at each check.
Indeed, each component of a disconnected pairing is itself a kernel support.
If all components are stabilizers, so is their sum; otherwise a smaller nontrivial component contradicts minimality.
Similarly, a minimum nontrivial support cannot strictly contain a nonzero stabilizer support, since adding that stabilizer removes its support and decreases the weight.
These observations justify restricting the search to minimal connected supports and stopping a branch when its current support is already a stabilizer.

The HGP check rows need not admit the three-group partition used by a colored-matching enumeration.
Instead, the search starts at a variable and recursively chooses an unsatisfied check, branching over the additional columns that can make its intersection with the support even.
An exclusion set prevents duplicate sibling branches without omitting any admissible extension.
If the number of unsatisfied checks exceeds three times the number of remaining columns allowed by the weight limit, the branch cannot close and is discarded.
A zero-syndrome support is tested against the stabilizer row space; a stabilizer closure terminates the branch by the minimality argument above.
This enumeration covers minimum nontrivial supports without requiring a three-group partition of the checks.

The simultaneous cyclic shift of all lift indices preserves both parity-check matrices up to row permutations.
It therefore preserves the corresponding stabilizer row spaces as well as the kernels.
The search uses one root representative per base column under this verified symmetry.
All candidate supports are checked on the full lifted matrices.
Only fully completed weight levels contribute to a lower bound; a time-limited partial search contributes none.
For this $P=64$ instance, all even weights through $16$ were completed on both sides using $450$ root representatives. No nontrivial logical support was found. The weight-$18$ searches were interrupted and are not used in the lower bound. Thus $d_X,d_Z\ge18$.
The ancillary distance records include the input coefficients, source code, complete pattern lists where applicable, completed-weight logs, and file hashes.
The lower bounds are computer-assisted exhaustive-search results.
The explicit upper-bound representatives can be verified separately by syndrome and stabilizer checks.